\documentclass[letterpaper,12pt]{article}
\usepackage{amssymb, amsmath}
\usepackage{amsthm}
\usepackage{algorithm}
\usepackage{algorithmic}
\usepackage{enumerate}
\usepackage[numbers,square,comma,sort&compress]{natbib}
\usepackage[pdftex,colorlinks]{hyperref}
\usepackage{a4}

\newtheorem{theorem}{Theorem}
\newtheorem{fact}[theorem]{Fact}
\newtheorem{lemma}[theorem]{Lemma}

\newtheorem*{markov}{Markov's inequality}
\newtheorem*{chebyshev}{Chebyshev's inequality}
\newcommand{\comment}[1]{}
 
\newcommand\bs[1]{\boldsymbol{#1}}

\begin{document}
\title{Metric random matchings with applications\footnote{A preliminary version of this paper
appears in \url{https://arxiv.org/abs/1702.03106}.}}

\author{
Ching-Lueh Chang \footnote{Department of Computer Science and
Engineering,
Yuan Ze University, Taoyuan, Taiwan. Email:
clchang@saturn.yzu.edu.tw}
\footnote{Supported
in part by the Ministry of Science and Technology of Taiwan under
grant 105-2221-E-155-047-.}
}


\maketitle

\begin{abstract}
Let $(\{1,2,\ldots,n\},d)$ be a metric space.
We analyze the expected value and the variance
of
$\sum_{i=1}^{\lfloor n/2\rfloor}\,d(\bs{\pi}(2i-1),\bs{\pi}(2i))$
for a uniformly random permutation $\bs{\pi}$ of $\{1,2,\ldots,n\}$,
leading to the following results:
\begin{itemize}
\item Consider the problem of finding
a point in $\{1,2,\ldots,n\}$ with the minimum
sum of distances
to all points.
We show that
this problem
has a randomized algorithm that
(1)~{\em always} outputs a $(2+\epsilon)$-approximate solution
in
expected $O(n/\epsilon^2)$ time
and that
(2)~inherits Indyk's~\cite{Ind99, Ind00} algorithm
to output a $(1+\epsilon)$-approximate solution in $O(n/\epsilon^2)$ time
with probability $\Omega(1)$,
where $\epsilon\in(0,1)$.
\item The average distance in $(\{1,2,\ldots,n\},d)$
can be approximated in $O(n/\epsilon)$ time
to within a multiplicative factor in $[\,1/2-\epsilon,1\,]$
with probability $1/2+\Omega(1)$,
where
$\epsilon>0$.
\item
Assume $d$ to be a graph metric.
Then the
average distance in $(\{1,2,\ldots,n\},d)$
can be approximated in $O(n)$ time
to within a multiplicative factor in $[\,1-\epsilon,1+\epsilon\,]$
with probability $1/2+\Omega(1)$,
where $\epsilon=\omega(1/n^{1/4})$.
\end{itemize}
\comment{ original abstract 20170324 20:01
Given
an $n$-point
metric
space,
consider the problem of finding
a point with the minimum
sum of distances
to all points.
We show that
this problem
has a randomized algorithm that
(1)~{\em always} outputs a $(2+\epsilon)$-approximate solution
in an expected $O(n/\epsilon^2)$ time
and that
(2)~inherits Indyk's~\cite{Ind99, Ind00} algorithm
to output a $(1+\epsilon)$-approximate solution in $O(n/\epsilon^2)$ time
with probability $\Omega(1)$,
for
all $\epsilon=1/n^{o(1)}$.
\comment{ 
Inheriting Indyk's~\cite{Ind00} algorithm,
our
algorithm
outputs
a $(1+\epsilon)$-approximate
$1$-median
in $O(n/\epsilon^2)$ time
with probability $\Omega(1)$.
}
As a by-product,
we give a Monte-Carlo $O(n)$-time algorithm
that approximates
the average distance of any graph metric to within a multiplicative factor
in $[\,1-\epsilon,1+\epsilon\,]$, for all $\epsilon=1/n^{o(1)}$.
}
\comment{ 
As a by-product,
we present a Monte-Carlo $O(n/\epsilon^2)$-time algorithm for estimating
the sum of all distances in an $n$-point metric space
to within
a ratio in $[\,1-\epsilon,1+\epsilon\,]$, for each constant $\epsilon>0$.
}
\end{abstract}

\section{Introduction}

A metric space is a nonempty set $M$ endowed with a
metric,
i.e.,
a function
$d\colon M\times M\to[\,0,\infty\,)$
such that
\begin{itemize}
\item $d(x,y)=0$ if and only if $x=y$ (identity of indiscernibles),
\item $d(x,y)=d(y,x)$ (symmetry), and
\item $d(x,y)+d(y,z)\ge d(x,z)$ (triangle inequality)
\end{itemize}
for all
$x$, $y$, $z\in M$~\cite{Rud76}.

For all $n\in\mathbb{Z}^+$, define $[n]\equiv\{1,2,\ldots,n\}$.
Given
$n\in\mathbb{Z}^+$
and oracle access to a metric $d\colon [n]\times[n]\to[\,0,\infty\,)$,
{\sc metric $1$-median}
asks for
$\mathop{\mathrm{argmin}}_{y\in[n]}\,\sum_{x\in[n]}\,d(y,x)$,
breaking ties arbitrarily.
It generalizes the classical median selection on the real line
and has a
brute-force
$\Theta(n^2)$-time algorithm.
More generally, {\sc metric $k$-median} asks for
$c_1$, $c_2$, $\ldots$,
$c_k\in[n]$
minimizing
$\sum_{x\in[n]}\,\min_{i=1}^k\,d(x,c_i)$.
Because
$d(\cdot,\cdot)$ defines
$\binom{n}{2}=\Theta(n^2)$ nonzero distances,
only $o(n^2)$-time algorithms are said to
run in
sublinear time~\cite{Ind99}.
For all $\alpha\ge1$,
an
$\alpha$-approximate
$1$-median
is a point $p\in[n]$
satisfying
$$\sum_{x\in[n]}\,d\left(p,x\right)\le\alpha\cdot
\min_{y\in[n]}\,\sum_{x\in[n]}\,d\left(y,x\right).$$

For all $\epsilon>0$,
{\sc metric $1$-median}
has
a Monte Carlo $(1+\epsilon)$-approximation
$O(n/\epsilon^2)$-time
algorithm~\cite{Ind99, Ind00}.
Guha et al.~\cite{GMMMO03}
show that {\sc metric $k$-median}
has
a Monte Carlo, $O(\exp(O(1/\epsilon)))$-approximation,
$O(nk\log n)$-time, $O(n^{\epsilon})$-space and one-pass algorithm for all small $k$
as well as a deterministic, $O(\exp(O(1/\epsilon)))$-approximation,
$O(n^{1+\epsilon})$-time, $O(n^{\epsilon})$-space and one-pass algorithm.
Given
$n$ points in $\mathbb{R}^D$ with $D\ge 1$,
the Monte Carlo algorithms of Kumar et al.~\cite{KSS10}
find
a $(1+\epsilon)$-approximate $1$-median
in
$O(D\cdot\exp(1/\epsilon^{O(1)}))$ time
and a $(1+\epsilon)$-approximate solution to {\sc metric $k$-median}
in $O(Dn\cdot\exp((k/\epsilon)^{O(1)}))$ time.
All randomized $O(1)$-approximation algorithms for {\sc metric $k$-median}
take $\Omega(nk)$ time~\cite{MP04, GMMMO03}.
Chang~\cite{Cha15}
shows that {\sc metric $1$-median} has
a deterministic, $(2h)$-approximation, $O(hn^{1+1/h})$-time
and nonadaptive algorithm for
all
constants
$h\in\mathbb{Z}^+\setminus\{1\}$, generalizing the results of Chang~\cite{Cha13} and
Wu~\cite{Wu14}.
On the other hand,
he
disproves
the
existence
of
deterministic $(2h-\epsilon)$-approximation $O(n^{1+1/(h-1)}/h)$-time
algorithms
for all constants $h\in\mathbb{Z}^+\setminus\{1\}$ and $\epsilon>0$~\cite{Cha16COCOON, Cha17}.

In social network analysis, the closeness centrality of a point $v$
is
the reciprocal of the
average distance
from $v$ to all
points~\cite{WF94}.
So {\sc metric $1$-median}
asks for
a point with the maximum closeness
centrality.
Given oracle access to a graph metric,
the Monte-Carlo algorithms of
Goldreich and Ron~\cite{GR08} and Eppstein and Wang~\cite{EW04}
estimate the closeness centrality of a given point and those of all points, respectively.

All known
sublinear-time
algorithms
for {\sc metric $1$-median}
are
either deterministic or
Monte Carlo,
the latter having
a positive probability of failure.
For example, Indyk's Monte Carlo $(1+\epsilon)$-approximation algorithm
outputs
with a positive probability
a solution
without approximation guarantees.
In contrast,
we show
that {\sc metric $1$-median} has
a randomized
algorithm
that {\em always} outputs a
$(2+\epsilon)$-approximate solution
in
expected $O(n/\epsilon^2)$
time
for all
$\epsilon\in(0,1)$.
So,
excluding
the
known
deterministic algorithms (which are Las Vegas only in the degenerate
sense),
this paper gives
the {\em first} Las Vegas approximation algorithm for {\sc metric $1$-median}
with an expected
sublinear
running
time.
Note that
deterministic
sublinear-time
algorithms for {\sc metric $1$-median}
can be $4$-approximate but not $(4-\epsilon)$-approximate for any constant
$\epsilon>0$~\cite{Cha13, Cha17}.
So our
approximation ratio
of $2+\epsilon$
beats
that of
any
deterministic
sublinear-time
algorithm.
Inheriting
Indyk's algorithm,
our algorithm
outputs a $(1+\epsilon)$-approximate $1$-median in
$O(n/\epsilon^2)$ time with probability $\Omega(1)$ for all
$\epsilon\in(0,1)$.
\comment{ 
In case our algorithm
fails to
output a
$(1+\epsilon)$-approximate $1$-median, it
nonetheless
outputs a $(2+\epsilon)$-approximate
$1$-median.
}

\comment{ 
As a by-product of our derivations,
we present a Monte Carlo $O(n/\epsilon^2)$-time algorithm for estimating $\sum_{u,v\in [n]}\,d(u,v)$
to within an additive error of $\epsilon\cdot\sum_{u,v\in [n]}\,d(u,v)$
with an $\Omega(1)$ probability of success, for each constant $\epsilon>0$.
Previously, the best algorithm for such estimation
needs
$O(n/\epsilon^{7/2})$ time~\cite[Sec.~8]{Ind99}.
}

\comment{ 
So
we have
the first Las Vegas
approximation
algorithm
for {\sc metric $1$-median}
with an expected sublinear running time.
Because
the best approximation ratio
achievable by deterministic sublinear-time algorithms for {\sc metric $1$-median}
is $4$~\cite{Cha13, Cha15},
Las Vegas
approximation
algorithms
have a better approximation ratio.
}

Indyk~\cite{Ind99, Ind00} gives a Monte-Carlo $O(n/\epsilon^{3.5})$-time
algorithm
that approximates the average distance in any
metric space
$([n],d)$
to within a multiplicative factor in $[\,1-\epsilon,1+\epsilon\,]$,
for all
$\epsilon>0$.
Barhum, Goldreich and Shraibman~\cite{BGS07}
improve
Indyk's
time complexity of $O(n/\epsilon^{3.5})$
to
$O(n/\epsilon^2)$.
This paper gives
a Monte-Carlo
$O(n/\epsilon)$-time algorithm
that approximates the average distance
in
$([n],d)$
to within a multiplicative factor in $[\,1/2-\epsilon,1\,]$,
for all
$\epsilon>0$.
For all $\epsilon=\omega(1/n^{1/4})$,
we
present
a Monte-Carlo
$O(n)$-time
algorithm
approximating
the average distance
of any graph metric
to within a multiplicative factor in $[\,1-\epsilon,1+\epsilon\,]$.
But for general metrics, we do not know whether
the
$O(n/\epsilon^2)$
running time
of Barhum, Goldreich and Shraibman
can be improved to $O(n/\epsilon^{2-\Omega(1)})$.

\section{Definitions and preliminaries}

For a metric space $([n],d)$,
\begin{eqnarray}
\bar{r}&\equiv& \frac{1}{n^2}\cdot\sum_{x,y\in[n]}\,d\left(x,y\right),
\label{averagedistance}\\
p^*
&\equiv&
\mathop{\mathrm{argmin}}_{p\in[n]}\,\sum_{x\in [n]}\,
d\left(p,x\right),\label{optimalpoint}
\end{eqnarray}
breaking ties arbitrarily
in equation~(\ref{optimalpoint}).
So
$\bar{r}$ is the average distance in $([n],d)$, and
$p^*$ is a $1$-median.
\comment{ 
For brevity,
$$
B^2\left(x,R\right)\equiv B\left(x,R\right)\times B\left(x,R\right).
$$
The
pairs in $B^2(x,R)$ are ordered.
}

An algorithm $A$ with oracle access to
$d\colon [n]\times[n]\to[\,0,\infty\,)$
is denoted by $A^d$
and
may query $d$ on any $(x,y)\in[n]\times[n]$ for $d(x,y)$.
In this paper, all
Landau symbols (such as $O(\cdot)$, $o(\cdot)$, $\Omega(\cdot)$ and $\omega(\cdot)$)
are w.r.t.\ $n$.
The following
result
is due to Indyk.

\begin{fact}[\cite{Ind99, Ind00}]\label{Indykfact}
For all $\epsilon>0$,
{\sc metric $1$-median} has a Monte Carlo $(1+\epsilon)$-approximation
$O(n/\epsilon^2)$-time algorithm
with
a failure probability of at most
$1/e$.
\end{fact}

Henceforth,
denote Indyk's algorithm in Fact~\ref{Indykfact} by {\sf Indyk median}.
It is given
$n\in\mathbb{Z}^+$,
$\epsilon>0$
and oracle access to
a metric
$d\colon [n]\times[n]\to[\,0,\infty\,)$.
The following fact
on estimating the average distance
is due to Barhum, Goldreich and Shraibman.

\begin{fact}[\cite{BGS07}]\label{averagedistancepriorresult}
Given $n\in\mathbb{Z}^+$, $\epsilon>0$ and oracle access to a metric $d\colon[n]\times[n]
\to[\,0,\infty\,)$,
a real number in
$\left[\,(1-\epsilon)\bar{r},(1+\epsilon)\bar{r}\,\right]$
can be found in $O(n/\epsilon^2)$ time with probability at least $1/2+\Omega(1)$.
\end{fact}

\comment{ 
\begin{fact}[Implicit in~{\cite[Theorem~22]{Cha12}}]\label{uniformlyrandompointisgood}
A uniformly random point of $[n]$ is a $4$-approximate $1$-median
with probability at least $1/2$.
\end{fact}
}

\comment{ 
\begin{markov}
[\cite{MR95}]
For a non-negative random variable $X$ and $a>0$,
$$
\Pr
\left|\,
X\ge a
\,\right]\le \frac{\mathop{\mathrm{E}}[X]}{a}.
$$
\end{markov}
}

\begin{chebyshev}
[\cite{MR95}]
Let
$X$
be a random variable
with a finite expected value and a finite
nonzero variance.
Then for all
$k\ge1$,
$$
\Pr\left[\,
\left|\,
X-\mathop{\mathrm E}[X]\,\right|\ge k\sqrt{\mathop{\mathrm{var}}(X)}
\,\right]\le \frac{1}{k^2}.
$$
\end{chebyshev}

\section{Las Vegas approximation
for metric $1$-median
selection}\label{mainresultsection}

\comment{ 
Clearly,
\begin{eqnarray}
\sum_{x\in[n]}\,d\left(p^*,x\right)
\stackrel{\text{(\ref{optimalpoint})}}{\le}
\frac{1}{n}\cdot\sum_{p\in[n]}\,\sum_{x\in[n]}\,d\left(p,x\right)
\stackrel{\text{(\ref{averagedistance})}}{=}
nr.
\end{eqnarray}
}

This section
presents
a randomized algorithm that {\em always} outputs a
$(2+\epsilon)$-approximate $1$-median, where $\epsilon\in(0,1)$.
Clearly,
\begin{eqnarray}
\sum_{x\in[n]}\,d\left(p^*,x\right)
\stackrel{\text{(\ref{optimalpoint})}}{=}
\min_{p\in[n]}\,\sum_{x\in[n]}\,d\left(p,x\right)
\le
\frac{1}{n}\cdot \sum_{p\in[n]}\,\sum_{x\in[n]}\,d\left(p,x\right)
\stackrel{\text{(\ref{averagedistance})}}{=}
n\bar{r}.\label{bestisnoworsethanaverage}
\end{eqnarray}

For each permutation $\pi\colon[n]\to[n]$,
\begin{eqnarray}
\sum_{i=1}^{\lfloor n/2\rfloor}\,
d\left(\pi\left(2i-1\right),\pi\left(2i\right)\right)
\le
\sum_{i=1}^{\lfloor n/2\rfloor}\,
d\left(p^*,\pi\left(2i-1\right)\right)+d\left(p^*,\pi\left(2i\right)\right)
\le \sum_{x\in[n]}\,d\left(p^*,x\right),
\label{matchingsizeandoptimalsumofdistances}
\end{eqnarray}
where the first and the second inequalities follow from
the triangle inequality and
the injectivity of $\pi$.

\begin{lemma}\label{lemmaforapproximationratio}
When line~5 of {\sf Las Vegas median} in Fig.~\ref{lasvegasalgorithm}
is run,
$z$ is
a $(2+\epsilon)$-approximate $1$-median.
\comment{ 
For
all
$y\in[n]$ and
permutations $\pi$, if
$$
\sum_{x\in[n]}\,d\left(y,x\right)
\le \left(2+\epsilon\right)
\sum_{i=1}^{\lfloor n/2\rfloor}\,
d\left(\pi\left(2i-1\right),\pi\left(2i\right)\right),
$$
then $y$ is a $(2+\epsilon)$-approximate $1$-median.
}
\end{lemma}
\begin{proof}
The
condition in line~4 of {\sf Las Vegas median} implies
$$
\sum_{x\in[n]}\,
d\left(z,x\right)
\stackrel{\text{(\ref{matchingsizeandoptimalsumofdistances})}}{\leq}
\left(2+\epsilon\right)
\sum_{x\in[n]}\,
d\left(p^*,x\right)
\stackrel{\text{(\ref{optimalpoint})}}{=}
\left(2+\epsilon\right)
\min_{p\in[n]}\,\sum_{x\in[n]}\,
d\left(p,x\right).
$$
So
when
line~5 is run, it returns a $(2+\epsilon)$-approximate $1$-median.
\end{proof}


\begin{figure}
\begin{algorithmic}[1]
\WHILE{\text{\sf true}}
  \STATE $z\leftarrow\text{\sf Indyk median}^d(n,\epsilon/8)$;
  \STATE Pick independent and uniformly random permutations
$\bs{\pi}_1$, $\bs{\pi}_2$, $\ldots$,
$\bs{\pi}_{80\lceil 1/\epsilon\rceil}\colon[n]\to[n]$;
  \IF{there exists $j\in[\lceil 1/\epsilon\rceil]$
satisfying $\sum_{x\in[n]}\,d(z,x)\le (2+\epsilon)\sum_{i=1}^{\lfloor n/2\rfloor}\,d(\bs{\pi}_j(2i-1),\bs{\pi}_j(2i))$}
    \RETURN $z$;
  \ENDIF
\ENDWHILE
\end{algorithmic}
\caption{Algorithm {\sf Las Vegas median} with oracle access to
a metric $d\colon [n]\times[n]\to[\,0,\infty\,)$ and with inputs $n\in\mathbb{Z}^+$
and $\epsilon\in(0,1)$}
\label{lasvegasalgorithm}
\end{figure}

Inequalities~(\ref{bestisnoworsethanaverage})--(\ref{matchingsizeandoptimalsumofdistances})
yield the following.

\begin{lemma}\label{maximummatchingsize}
For each permutation $\pi\colon[n]\to[n]$,
$$
\sum_{i=1}^{\lfloor n/2\rfloor}\,
d\left(\pi\left(2i-1\right),\pi\left(2i\right)\right)
\le n\bar{r}.
$$
\end{lemma}
\comment{ 
\begin{proof}
We have
\begin{eqnarray*}
\sum_{i=1}^{\lfloor n/2\rfloor}\,
d\left(\pi\left(2i-1\right),\pi\left(2i\right)\right)
\stackrel{\text{(\ref{optimalpoint})--(\ref{matchingsizeandoptimalsumofdistances})}}{\le}
\frac{1}{n}\cdot\sum_{p\in[n]}\,\sum_{x\in[n]}\,d\left(p,x\right)
\stackrel{\text{(\ref{averagedistance})}}{=}
nr.
\end{eqnarray*}
\end{proof}
}

\begin{lemma}\label{expectedmatchingsize}
For a uniformly random permutation $\bs{\pi}\colon [n]\to[n]$,
$$
\mathop{\mathrm{E}}_{\bs{\pi}}
\left[\,\sum_{i=1}^{\lfloor n/2\rfloor}\,
d\left(\bs{\pi}\left(2i-1\right),\bs{\pi}\left(2i\right)\right)
\right]
=\left\lfloor\frac{n}{2}\right\rfloor\cdot\frac{n\bar{r}}{n-1}.
$$
\end{lemma}
\begin{proof}
For each $i\in[\lfloor n/2\rfloor]$, $\{\bs{\pi}(2i-1),\bs{\pi}(2i)\}$
is a uniformly random size-$2$ subset of $[n]$,
implying
\begin{eqnarray*}
\mathop{\mathrm{E}}_{\bs{\pi}}
\left[\,
d\left(\bs{\pi}\left(2i-1\right),\bs{\pi}\left(2i\right)\right)
\right]
&=&\frac{1}{n\cdot(n-1)}\cdot\sum_{\text{distinct $x$, $y\in[n]$}}\,
d\left(x,y\right)\\
&=&\frac{1}{n\cdot(n-1)}\cdot\sum_{x,y\in[n]}\,d\left(x,y\right)\\
&\stackrel{\text{(\ref{averagedistance})}}{=}&
\frac{n\bar{r}}{n-1},
\end{eqnarray*}
where the second equality follows from the identity of indiscernibles.
Finally, use the linearity of expectation.
\end{proof}

\comment{ 
\begin{lemma}\label{probabilitythatarandommatchingislarge}
For a uniformly random permutation $\bs{\pi}\colon [n]\to[n]$
and all sufficiently large $n$,
\begin{eqnarray}
\Pr_{\bs{\pi}}
\left[\,\sum_{i=1}^{\lfloor n/2\rfloor}\,
d\left(\bs{\pi}\left(2i-1\right),\bs{\pi}\left(2i\right)\right)
\ge\left(\frac{1}{2}-\frac{\epsilon}{8}\right)nr
\right]
\ge \frac{\epsilon}{8}.
\label{probabilityofasmallmatching}
\end{eqnarray}
\end{lemma}
\begin{proof}
Denote the left-hand side of inequality~(\ref{probabilityofasmallmatching})
by $p$.
Then by Lemma~\ref{maximummatchingsize},
$$
\mathop{\mathrm{E}}_{\bs{\pi}}
\left[\,
\sum_{i=1}^{\lfloor n/2\rfloor}\,
d\left(\bs{\pi}\left(2i-1\right),\bs{\pi}\left(2i\right)\right)
\,\right]
\le
p nr
+
\left(1-p\right)\left(\frac{1}{2}-\frac{\epsilon}{8}\right)nr.
$$
This and Lemma~\ref{expectedmatchingsize} complete the proof.
\end{proof}
}

\begin{lemma}\label{probabilitythatoneoftherandommatchingsislarge}
For all $\epsilon\in(0,1)$ and in
each iteration of the {\bf while} loop of {\sf Las Vegas median},
\begin{eqnarray}
\Pr
\left[\,
\exists j\in\left[80\cdot\left\lceil\frac{1}{\epsilon}\right\rceil\right],\,
\sum_{i=1}^{\lfloor n/2\rfloor}\,
d\left(\bs{\pi}_j\left(2i-1\right),\bs{\pi}_j\left(2i\right)\right)
\ge\left(\frac{1}{2}-\frac{\epsilon}{8}\right)n\bar{r}
\right]
\ge0.9,
\label{probabilityofalargematching}
\end{eqnarray}
where the probability is taken over $\bs{\pi}_1$, $\bs{\pi}_2$,
$\ldots$, $\bs{\pi}_{80\lceil 1/\epsilon\rceil}$ in line~3 of {\sf Las Vegas median}.
\end{lemma}
\begin{proof}
Let
$\bs{\pi}\colon[n]\to[n]$ be a uniformly random permutation
and
\begin{eqnarray}
\alpha=\Pr_{\bs{\pi}}
\left[\,\sum_{i=1}^{\lfloor n/2\rfloor}\,
d\left(\bs{\pi}\left(2i-1\right),\bs{\pi}\left(2i\right)\right)
\ge\left(\frac{1}{2}-\frac{\epsilon}{8}\right)n\bar{r}
\right].
\end{eqnarray}
So
by
Lemma~\ref{maximummatchingsize},
$$
\mathop{\mathrm{E}}_{\bs{\pi}}
\left[\,
\sum_{i=1}^{\lfloor n/2\rfloor}\,
d\left(\bs{\pi}\left(2i-1\right),\bs{\pi}\left(2i\right)\right)
\,\right]
\le
\alpha n\bar{r}
+
\left(1-\alpha\right)\left(\frac{1}{2}-\frac{\epsilon}{8}\right)
n\bar{r}.
$$
This and Lemma~\ref{expectedmatchingsize}
imply $\alpha\ge \epsilon/8$.
So the left-hand side of inequality~(\ref{probabilityofalargematching})
is at least
$1-(1-\epsilon/8)^{80\lceil 1/\epsilon\rceil}\ge0.9$.
\end{proof}

\comment{ 
\begin{lemma}
For each $y\in[n]$ and a uniformly random permutation $\bs{\pi}$, if
$$
\sum_{x\in[n]}\,d\left(y,x\right)
\le \left(2+\epsilon\right)
\sum_{i=1}^{\lfloor n/2\rfloor}\,
d\left(\bs{\pi}\left(2i-1\right),\bs{\pi}\left(2i\right)\right),
$$
then $y$ is a $(2+\epsilon)$-approximate $1$-median.
\end{lemma}
}
\comment{ 
For all $y\in[n]$,
$$\sum_{x\in[n]}\,d\left(y,x\right)
\ge
\sum_{i=1}^{\lfloor n/2\rfloor}\,
d\left(y,\bs{\pi}\left(2i-1\right)\right)
+d\left(y,\bs{\pi}\left(2i\right)\right)
\ge
\sum_{i=1}^{\lfloor n/2\rfloor}\,
d\left(\bs{\pi}\left(2i-1\right), \bs{\pi}\left(2i\right)\right),
$$
where the first and the second inequalities follow from
the injectivity of $\bs{\pi}$ and the triangle inequality.
So
inequalities~(\ref{qualityofanearoptimalsolution})--(\ref{lowerboundofarandommatchingsize})
imply the approximation ratio of $z$
to be at most
$$
\frac{1+\epsilon/8}{}
$$
}

\comment{ 
\begin{lemma}\label{probabilitythatIndykalgorithmgivesagoodsolution}
In each execution of line~2 of {\sf Las Vegas median},
$$
\Pr\left[\,
\sum_{x\in[n]}\,d\left(z,x\right)
\le \left(1+\frac{\epsilon}{8}\right)
nr
\,\right]
=\Omega(1),
$$
where the probability is taken over the random coin tosses of {\sf Indyk median}.
\end{lemma}
\begin{proof}
By Fact~\ref{Indykfact},
$$
\Pr\left[\,
\sum_{x\in[n]}\,d\left(z,x\right)
\le \left(1+\frac{\epsilon}{8}\right)
\sum_{x\in[n]}\,d\left(p^*,x\right)
\,\right]
=\Omega(1).
$$
This and inequality~(\ref{bestisnoworsethanaverage})
complete the proof.
\end{proof}
}

\begin{lemma}\label{probabilitythatIndykalgorithmisgoodandmatchingislarge}
For all $\epsilon\in(0,1)$ and in
each iteration of the {\bf while} loop of {\sf Las Vegas median},
\begin{eqnarray}
&&\Pr\left[\,
\left(
\sum_{x\in[n]}\,d\left(z,x\right)
\le \left(1+\frac{\epsilon}{8}\right)
\min_{p\in[n]}\,\sum_{x\in[n]}\,d\left(p,x\right)
\right)\right.\nonumber\\
&&\left.
\land\left(
\exists j\in\left[80\cdot\left\lceil\frac{1}{\epsilon}\right\rceil\right],\,
\sum_{i=1}^{\lfloor n/2\rfloor}\,
d\left(\bs{\pi}_j\left(2i-1\right),\bs{\pi}_j\left(2i\right)\right)
\ge \left(\frac{1}{2}-\frac{\epsilon}{8}\right)n\bar{r}
\right)
\right.\nonumber\\
&&\left.
\land
\left(\exists j\in\left[80\cdot\left\lceil\frac{1}{\epsilon}\right\rceil\right],\,
\sum_{x\in[n]}\,d\left(z,x\right)
\le \left(2+\epsilon\right)\sum_{i=1}^{\lfloor n/2\rfloor}\,
d\left(\bs{\pi}_j\left(2i-1\right),\bs{\pi}_j\left(2i\right)\right)
\right)
\,\right]\nonumber\\
&=&\frac{1}{2}+\Omega(1),\label{ultimateequation}
\end{eqnarray}
where the probability is taken over $\bs{\pi}_1$,
$\bs{\pi}_2$, $\ldots$,
$\bs{\pi}_{80\lceil 1/\epsilon\rceil}$
and the random
coin tosses of {\sf Indyk median}.
\end{lemma}
\begin{proof}
By
Fact~\ref{Indykfact} and line~2 of {\sf Las Vegas median},
the first condition within $\Pr[\cdot]$ in equation~(\ref{ultimateequation})
holds with probability
at least $1-1/e$
over
the
random
coin tosses of
{\sf Indyk median}.
By
Lemma~\ref{probabilitythatoneoftherandommatchingsislarge},
\comment{ 
\begin{eqnarray*}
\sum_{x\in[n]}\,d\left(z,x\right)
&\le& \left(1+\frac{\epsilon}{8}\right)
nr,\\
\exists j\in\left[\frac{1}{\epsilon}\right],\,
\sum_{i=1}^{\lfloor n/2\rfloor}\,
d\left(\bs{\pi}_j\left(2i-1\right),\bs{\pi}_j\left(2i\right)\right)
&\ge& \left(\frac{1}{2}-\frac{\epsilon}{8}\right)nr
\end{eqnarray*}
}
the
second condition holds
with probability
at least $0.9$
over $\bs{\pi}_1$, $\bs{\pi}_2$,
$\ldots$, $\bs{\pi}_{80\lceil 1/\epsilon\rceil}$.
In summary, the first two conditions hold simultaneously
with probability at least $(1-1/e)\cdot 0.9=1/2+\Omega(1)$
(note
that the random coin tosses of {\sf Indyk median}
are independent of $\bs{\pi}_1$, $\bs{\pi}_2$,
$\ldots$, $\bs{\pi}_{80\lceil 1/\epsilon\rceil}$).
Finally, the first two conditions together imply the third
by inequality~(\ref{bestisnoworsethanaverage}) and the easy fact that
$$
\left(1+\frac{\epsilon}{8}\right)\le\left(2+\epsilon\right)
\left(\frac{1}{2}-\frac{\epsilon}{8}\right).
$$
\end{proof}

\comment{ 
\begin{proof}
By Fact~\ref{Indykfact},
$$
\Pr\left[\,
\sum_{x\in[n]}\,d\left(z,x\right)
\le \left(1+\frac{\epsilon}{8}\right)
\sum_{x\in[n]}\,d\left(p^*,x\right)
\,\right]
=\Omega(1),
$$
where the probability is taken over the random coin tosses of {\sf Indyk median}.
This and inequality~(\ref{bestisnoworsethanaverage}) give
$$
\Pr\left[\,
\sum_{x\in[n]}\,d\left(z,x\right)
\le \left(1+\frac{\epsilon}{8}\right)
nr
\,\right]
=\Omega(1),
$$
which together with Lemma~\ref{probabilitythatarandommatchingislarge}
completes the proof (note that $\bs{\pi}$ is independent of the
random coin tosses of {\sf Indyk median}).
\end{proof}
}


\begin{theorem}\label{maintheorem}
For
all
$\epsilon\in(0,1)$,
{\sc metric $1$-median} has a randomized algorithm that
(1)~{\em always} outputs a $(2+\epsilon)$-approximate solution in
expected $O(n/\epsilon^2)$ time and (2)~outputs a $(1+\epsilon)$-approximate
solution in $O(n/\epsilon^2)$ time with probability $\Omega(1)$.
\end{theorem}
\begin{proof}
By
Lemma~\ref{probabilitythatIndykalgorithmisgoodandmatchingislarge},
each execution of lines~4--5 of {\sf Las Vegas median} returns with probability
$1/2+\Omega(1)$.
So the expected number of iterations is $O(1)$.
By Fact~\ref{Indykfact},
line~2 takes $O(n/\epsilon^2)$ time.
Line~3 takes $80\lceil 1/\epsilon\rceil\cdot O(n)$ time by the Knuth shuffle.
Clearly,
lines~4--5 take $O(n/\epsilon)$ time.
In summary, the expected running time of {\sf Las Vegas median} is
$O(1)\cdot O(n/\epsilon^2)=O(n/\epsilon^2)$.
To prevent {\sf Las Vegas median} from running forever, find a $1$-median
by brute force (which obviously takes $O(n^2)$ time) after $n^2$ steps of
computation.
By Lemma~\ref{lemmaforapproximationratio}, {\sf Las Vegas median}
is $(2+\epsilon)$-approximate.

By Lemma~\ref{probabilitythatIndykalgorithmisgoodandmatchingislarge},
$z$ is $(1+\epsilon/8)$-approximate and is also returned in line~5
with probability $\Omega(1)$ in
the first
(in fact, any)
iteration.
Finally, the previous paragraph has
shown each
iteration to take $O(n/\epsilon^2)$ time.
\end{proof}

\comment{ 
\begin{eqnarray}
\sum_{x\in[n]}\,d\left(z,x\right)
&\le& \left(1+\frac{\epsilon}{8}\right)
nr,\label{qualityofanearoptimalsolution}\\
\sum_{i=1}^{\lfloor n/2\rfloor}\,
d\left(\bs{\pi}\left(2i-1\right),\bs{\pi}\left(2i\right)\right)
&\ge& \left(\frac{1}{2}-\frac{\epsilon}{8}\right)nr,
\label{lowerboundofarandommatchingsize}
\end{eqnarray}
}

By
Fact~\ref{Indykfact},
{\sf Indyk median} satisfies condition~(2)
in Theorem~\ref{maintheorem}.
But it does not satisfy condition~(1).

We
now
justify
the optimality of
the ratio of $2+\epsilon$ in
Theorem~\ref{maintheorem}.
%
%
%
Let
$A$ be a randomized algorithm
that
always
outputs
a
$(2-\epsilon)$-approximate
$1$-median.
Furthermore, denote by $p\in [n]$ (resp., $Q\subseteq [n]\times [n]$)
the output (resp., the set of queries as unordered pairs)
of $A^{d_1}(n)$, where $d_1$ is the discrete metric (i.e.,
$d_1(x,y)=1$ and $d_1(x,x)=0$ for all distinct $x$, $y\in [n]$).
Without loss of generality, assume $(p,y)\in Q$ for all $y\in [n]\setminus\{p\}$ by adding dummy
queries.
So
the queries in $Q$ witness
that
\begin{eqnarray}
\sum_{y\in [n]\setminus\{p\}}\,d_1\left(p,y\right)=n-1.
\label{outputunderthediscretemetric}
\end{eqnarray}
Assume
without loss of generality
that $A$ never queries for the distance from a point to itself.

In the sequel,
consider the case
that $|Q|<\epsilon\cdot(n-1)^2/8$.
By
the averaging argument, there exists a point $\hat{p}\in [n]\setminus \{p\}$
involved in at most $2\cdot|Q|/(n-1)$ queries in $Q$ (note that each
query involves two points).
Because
every
function $f\colon [n]\times[n]\to[\,0,\infty\,)$
with
$$\left\{f\left(x,y\right)\mid \left(x,y\in[n]\right)\land\left(
x\neq y\right)\right\}\subseteq\left\{\frac{1}{2},1\right\}$$
satisfies the triangle inequality,
$A$ cannot exclude the possibility that $d_1(\hat{p},y)=1/2$ for all $y\in[n]\setminus\{\hat{p}\}$
satisfying $(\hat{p},y)\notin Q$.
In summary,
$A$ cannot rule out the case that
\begin{eqnarray}
\sum_{y\in[n]}\,d_1\left(\hat{p},y\right)&\le& \frac{2\cdot |Q|}{n-1}\cdot 1
+\left(n-1-\frac{2\cdot |Q|}{n-1}\right)\cdot \frac{1}{2}
< \left(\frac{1}{2}+\frac{\epsilon}{8}\right)\cdot(n-1).\,\,\,\,\,
\label{acasethatcannotberuledout}
\end{eqnarray}
Equations~(\ref{outputunderthediscretemetric})--(\ref{acasethatcannotberuledout})
contradict
the guarantee that $p$ is $(2-\epsilon)$-approximate.
Consequently, the case that $|Q|<\epsilon\cdot(n-1)^2/8$ should {\em never}
happen.
The next theorem summarizes the above.

\begin{theorem}
{\sc Metric $1$-median}
has no
randomized algorithm
that always outputs a $(2-\epsilon)$-approximate
solution and that makes
fewer than
$\epsilon\cdot (n-1)^2/8$ queries
with a positive probability
given oracle access to the discrete metric,
for any
constant
$\epsilon\in(0,1)$.
\end{theorem}

Lemmas~\ref{maximummatchingsize}~and~\ref{probabilitythatoneoftherandommatchingsislarge}
yield the following estimation of the average distance.

\begin{theorem}\label{theoremforaveragedistance}
Given
$n\in\mathbb{Z}^+$, $\epsilon>0$ and
oracle access to a metric $d\colon[n]\times[n]\to[\,0,\infty\,)$,
a real number in $[\,(1/2-\epsilon)\bar{r},\bar{r}\,]$
can be found
in $O(n/\epsilon)$ time
with probability $1/2+\Omega(1)$.
\end{theorem}
\begin{proof}
By
Lemmas~\ref{maximummatchingsize}~and~\ref{probabilitythatoneoftherandommatchingsislarge},
\begin{eqnarray}
\frac{1}{n}\cdot
\max_{j\in[80\cdot\lceil 1/\epsilon\rceil]}\,
\sum_{i=1}^{\lfloor n/2\rfloor}\,d\left(\bs{\pi}_j\left(2i-1\right),
\bs{\pi}_j\left(2i\right)\right)
\in\left[\,\left(\frac{1}{2}-\frac{\epsilon}{8}\right)\bar{r},
\bar{r}\,\right]
\label{rangeofestimationofaveragedistance}
\end{eqnarray}
with probability $1/2+\Omega(1)$.
The Knuth shuffle picks
$\bs{\pi}_1$, $\bs{\pi}_2$, $\ldots$, $\bs{\pi}_{80\lceil 1/\epsilon\rceil}$
in $80\lceil 1/\epsilon\rceil\cdot O(n)$
time.
Then
the left-hand side of
relation~(\ref{rangeofestimationofaveragedistance})
can be calculated in $O(n/\epsilon)$ time.
\end{proof}

Note that the estimation
of
the average distance
in
Theorem~\ref{theoremforaveragedistance}
has only
one-sided error.
The time complexity (resp., approximation ratio) in
Theorem~\ref{theoremforaveragedistance}
is better (resp., worse) than that in Fact~\ref{averagedistancepriorresult}.

\comment{ 
But we do not know whether the time complexity in
Theorem~\ref{theoremforaveragedistance}
and the approximation ratio in Fact~\ref{averagedistancepriorresult}
can be achieved simultaneously.
}

\section{Estimating the average distance of a graph
metric}\label{approximationratiosection}

Throughout
this
section,
take
any
$\epsilon=\omega(1/n^{1/4})$
less than a small constant,
e.g., $\epsilon=10^{-100}$.
Define
\begin{eqnarray}
\delta&\equiv& \frac{\epsilon^2}{10^{10}},\label{thedeltavalue}\\
r&\equiv&\frac{1}{n}\cdot\sum_{x\in[n]}\,d\left(p^*,x\right),
\label{theaveragedistancefromanoptimalpoint}
\end{eqnarray}
where $p^*$ is as in equation~(\ref{optimalpoint}).
As $\epsilon=\omega(1/n^{1/4})$,
$\delta=\omega(1/\sqrt{n})$ by equation~(\ref{thedeltavalue}).
\comment{ 
By line~1 of
{\sf average distance}
in Fig.~\ref{mainalgorithm},
$\delta>0$
is likewise
small,
and $\delta=1/n^{o(1)}$.
Furthermore, take $\bs{u}$, $r$ and $\bs{\pi}$
as in lines~2--4 of {\sf average distance}.
}

\comment{ 
The following lemma
implies
that
$z$ in line~3 of {\sf Las Vegas median}
is a solution (to {\sc metric $1$-median}) no worse than
those
in
$[n]\setminus B(z,8r)$, where $r$ is as in line~4.
}


\comment{ 
This section analyzes
the probability of running line~7
in
any
particular
iteration of the {\bf while} loop
of {\sf Las Vegas median}.
The following lemma
uses an easy averaging argument.
}

\comment{ 
Assume
$n$ to be sufficiently large
so that
$|B(z,\delta nr)|\ge 4$ by Lemma~\ref{thesmallradiusballislarge}.
}

\comment{ 
Define
\begin{eqnarray}
r'
\equiv
\frac{1}{|B(z,\delta nr)|^2}
\cdot
\sum_{u, v\in B(z,\delta nr)}\,
d\left(u,v\right)
\label{smallerballaveragedefinition}
\end{eqnarray}
to
be the average distance
in $B(z,\delta nr)$.
}

\comment{ 
As $z\in B(z,\delta nr)$, the
denominator in the right-hand side of
equation~(\ref{smallerballaveragedefinition})
is nonzero.
}

\begin{lemma}\label{inneraverageandoverallaverage}
$\bar{r}\leq 2r$.
\end{lemma}
\begin{proof}
By
equation~(\ref{averagedistance}) and
the triangle inequality,
\begin{eqnarray}
\bar{r}
&\le&
\frac{1}{n^2}
\cdot
\sum_{x, y\in [n]}\,
\left(d\left(p^*,x\right)+d\left(p^*,y\right)\right)
\nonumber\\
&=&
\frac{1}{n^2}
\cdot
n\cdot\left(
\sum_{x\in [n]}\,
d\left(p^*,x\right)
+\sum_{y\in [n]}\, d\left(p^*,y\right)
\right)\nonumber\\
&=&
\frac{2}{n}
\cdot
\sum_{x\in [n]}\,
d\left(p^*,x\right).
\label{needanumberhere}
\end{eqnarray}
Equations~(\ref{theaveragedistancefromanoptimalpoint})--(\ref{needanumberhere})
complete the proof.
\end{proof}

\begin{figure}
\begin{algorithmic}[1]
\STATE Pick a
uniformly random
permutation $\bs{\pi}\colon [n]\to [n]$;
\RETURN $\sum_{i=1}^{\lfloor n/2\rfloor}\,d(\bs{\pi}(2i-1),\bs{\pi}(2i))\cdot 2/n$;
\end{algorithmic}
\caption{Algorithm {\sf average distance} with oracle access to a metric
$d\colon [n]\times[n]\to[\,0,\infty\,)$
and with inputs $n\in\mathbb{Z}^+$
and
$\epsilon=\omega(1/n^{1/4})$.}
\label{mainalgorithm}
\end{figure}

As in line~1 of {\sf average distance} in Fig.~\ref{mainalgorithm},
let $\bs{\pi}\colon[n]\to[n]$ be a uniformly random permutation.
Clearly,
\begin{eqnarray}
&&\mathop{\mathrm E}_{\bs{\pi}}\left[\,
\left(\sum_{i=1}^{\lfloor n/2\rfloor}\,
d\left(\bs{\pi}\left(2i-1\right),\bs{\pi}\left(2i\right)\right)\right)^2
\,\right]
\label{startofthemeanofsquare}\\
&=&
\mathop{\mathrm E}_{\bs{\pi}}\left[\,
\sum_{i=1}^{\lfloor n/2\rfloor}\,
d\left(\bs{\pi}\left(2i-1\right),\bs{\pi}\left(2i\right)\right)
\cdot \sum_{j=1}^{\lfloor n/2\rfloor}\,
d\left(\bs{\pi}\left(2j-1\right),\bs{\pi}\left(2j\right)\right)
\,\right]\nonumber\\
&=&
\sum_{\text{distinct $i,j=1$}}^{\lfloor n/2\rfloor}\,
\mathop{\mathrm E}_{\bs{\pi}}\left[\,
d\left(\bs{\pi}\left(2i-1\right),\bs{\pi}\left(2i\right)\right)
\cdot d\left(\bs{\pi}\left(2j-1\right),\bs{\pi}\left(2j\right)\right)
\,\right]\nonumber\\
&+&
\sum_{i=1}^{\lfloor n/2\rfloor}\,
\mathop{\mathrm E}_{\bs{\pi}}\left[\,
d^2\left(\bs{\pi}\left(2i-1\right),\bs{\pi}\left(2i\right)\right)
\,\right],\label{endofthemeanofsquare}
\end{eqnarray}
where the last equality follows from the linearity of expectation
and the separation of pairs $(i,j)$
according to whether $i=j$.
The next three lemmas analyze the
variance
of
$$
\sum_{i=1}^{\lfloor n/2\rfloor}\,
d\left(\bs{\pi}\left(2i-1\right),\bs{\pi}\left(2i\right)\right).
$$

\begin{lemma}\label{distinctdistancesproduct}
{\small 
\begin{eqnarray*}
\sum_{\text{\rm distinct $i,j=1$}}^{\lfloor n/2\rfloor}\,
\mathop{\mathrm E}_{\bs{\pi}}\left[\,
d\left(\bs{\pi}\left(2i-1\right),\bs{\pi}\left(2i\right)\right)
\cdot d\left(\bs{\pi}\left(2j-1\right),\bs{\pi}\left(2j\right)\right)
\,\right]
\le\frac{1}{4}\cdot \left(1
+
O\left(\frac{1}{n}\right)
\right)n^2
\bar{r}^2.
\end{eqnarray*}
}
\end{lemma}
\begin{proof}
Pick any distinct $i$, $j\in[\,\lfloor n/2\rfloor\,]$.
Clearly,
$$\left\{\bs{\pi}\left(2i-1\right), \bs{\pi}\left(2i\right),
\bs{\pi}\left(2j-1\right),
\bs{\pi}\left(2j\right)\right\}$$
is a uniformly random size-$4$ subset of $[n]$.
So
\begin{eqnarray*}
&&\mathop{\mathrm E}_{\bs{\pi}}\left[\,
d\left(\bs{\pi}(2i-1),\bs{\pi}(2i)\right)
\cdot d\left(\bs{\pi}(2j-1),\bs{\pi}(2j)\right)
\,\right]\\
&=&
\frac{1}{n
\cdot(n-1)\cdot(n-2)
\cdot(n-3)}
\cdot
\sum_{\text{distinct $u$, $v$, $x$, $y\in [n]$}}\,
d\left(u,v\right)\cdot d\left(x,y\right).
\end{eqnarray*}

Clearly,
\begin{eqnarray*}
\sum_{\text{distinct $u$, $v$, $x$, $y\in [n]$}}\,
d\left(u,v\right)\cdot d\left(x,y\right)
&\le&
\sum_{u, v, x, y\in [n]}\,
d\left(u,v\right)\cdot d\left(x,y\right)\\
&=&
\sum_{u, v\in [n]}\,
d\left(u,v\right)
\cdot
\sum_{x, y\in [n]}\,
d\left(x,y\right)\\
&=&
\left(\sum_{x, y\in [n]}\,
d\left(x,y\right)
\right)^2.
\end{eqnarray*}
In summary,
\begin{eqnarray*}
&&
\sum_{\text{\rm distinct $i,j=1$}}^{\lfloor n/2\rfloor}\,
\mathop{\mathrm E}_{\bs{\pi}}\left[\,
d\left(\bs{\pi}\left(2i-1\right),\bs{\pi}\left(2i\right)\right)
\cdot d\left(\bs{\pi}\left(2j-1\right),\bs{\pi}\left(2j\right)\right)
\,\right]\\
&\le&
\left\lfloor\frac{n}{2}\right\rfloor
\left(\left\lfloor\frac{n}{2}\right\rfloor-1\right)
\cdot
\frac{1}{n
\cdot(n-1)\cdot(n-2)
\cdot(n-3)}
\cdot
\left(
\sum_{x, y\in [n]}\,
d\left(x,y\right)
\right)^2.
\end{eqnarray*}
This
and equation~(\ref{averagedistance})
complete
the proof.
\end{proof}

\comment{ 
Lemmas~\ref{thesmallradiusballislarge}~and~\ref{distinctdistancesproduct}
and equation~(\ref{smallerballaveragedefinition})
yield the following.

\begin{lemma}
\begin{eqnarray*}
\sum_{\text{\rm distinct $i,j=1$}}^{|B(z,\delta nr)|/2}\,
\mathop{\mathrm E}\left[\,
d\left(\pi(2i-1),\pi(2i)\right)
\cdot d\left(\pi(2j-1),\pi(2j)\right)
\,\right]
=\frac{1}{4}\cdot n^2\left(r'\right)^2\left(1\pm o(1)\right).
\end{eqnarray*}
\end{lemma}
}


Define
$$
\Delta\equiv\max_{x,y\in[n]}\,d(x,y)
$$
to be the diameter of $([n],d)$.

\begin{lemma}\label{thehardestparttoboundlemma}
If
\begin{eqnarray}
\delta nr \ge \Delta
\label{shallhaveanequationnumberfortheballequalsuniversething}
\end{eqnarray}
then
\begin{eqnarray}
\sum_{i=1}^{\lfloor n/2\rfloor}\,
\mathop{\mathrm E}_{\bs{\pi}}\left[\,
d^2\left(\bs{\pi}\left(2i-1\right),\bs{\pi}\left(2i\right)\right)
\,\right]
\le
\left(\frac{1}{2}+
O\left(\frac{1}{n}\right)
\right)\left(\delta n^2r\bar{r}+\delta^2 n r^2\right).
\label{boundforthedistancesquaressummed}
\end{eqnarray}
\end{lemma}
\begin{proof}
Clearly,
$\{\bs{\pi}(2i-1),\bs{\pi}(2i)\}$ is a uniformly random size-$2$ subset of
$[n]$ for each
$i\in[\,\lfloor n/2\rfloor\,]$.
Therefore,
\begin{eqnarray}
\sum_{i=1}^{\lfloor n/2\rfloor}\,
\mathop{\mathrm E}_{\bs{\pi}}\left[\,
d^2\left(\bs{\pi}\left(2i-1\right),\bs{\pi}\left(2i\right)\right)
\,\right]
&=&\sum_{i=1}^{\lfloor n/2\rfloor}\,
\frac{1}{n\cdot(n-1)}
\cdot
\sum_{\text{distinct $x$, $y\in [n]$}}\,d^2\left(x,y\right)
\,\,\,\,\,\,\,
\label{startofobjective}\\
&\le&\sum_{i=1}^{\lfloor n/2\rfloor}\,
\frac{1}{n\cdot(n-1)}
\cdot
\sum_{x, y\in [n]}\,d^2\left(x,y\right)
\nonumber\\
&=&\left\lfloor\frac{n}{2}\right\rfloor\cdot
\frac{1}{n\cdot(n-1)}
\cdot
\sum_{x, y\in [n]}\,d^2\left(x,y\right).
\nonumber
\end{eqnarray}
By inequality~(\ref{shallhaveanequationnumberfortheballequalsuniversething}),
\begin{eqnarray}
d\left(x,y\right)
\le
\delta nr
\label{constraint}
\end{eqnarray}
for all $x$, $y\in [n]$.

By equations~(\ref{averagedistance})~and~(\ref{startofobjective})--(\ref{constraint}),
the left-hand side of inequality~(\ref{boundforthedistancesquaressummed})
cannot exceed the optimal value of the following problem, called {\sc max square sum}:\\
\begin{quote}
Find
$d_{x,y}\in \mathbb{R}$
for all $x$, $y\in [n]$ to maximize
\begin{eqnarray}
\left\lfloor\frac{n}{2}\right\rfloor\cdot
\frac{1}{n\cdot(n-1)}
\cdot
\sum_{x, y\in [n]}\,d_{x,y}^2
\label{objectiveofoptimization}
\end{eqnarray}
subject to
\begin{eqnarray}
\frac{1}{n^2}\cdot
\sum_{x,y\in [n]}\, d_{x,y}
=
\bar{r},\label{averagedistanceconstraint}\\
\forall
x, y\in [n],\,\,
0\le
d_{x,y}
\le
\delta nr.\label{largestdistanceconstraint}
\end{eqnarray}
\end{quote}
Above, constraint~(\ref{averagedistanceconstraint})
(resp., (\ref{largestdistanceconstraint}))
mimics equation~(\ref{averagedistance})
(resp., inequality~(\ref{constraint}) and the
non-negativeness of distances).
Appendix~\ref{analyzingthemaximizationproblem}
bounds
the
optimal value of
{\sc max square sum}
from
above by
\begin{eqnarray}
\left\lfloor\frac{n}{2}\right\rfloor
\frac{1}{n\cdot(n-1)}
\cdot
\left(\left\lfloor\frac{n \bar{r}}{\delta r}\right\rfloor+1\right)
\cdot \left(\delta nr\right)^2.
\nonumber
\end{eqnarray}
This evaluates to be
at most
$$\left(\frac{1}{2}+O\left(\frac{1}{n}\right)\right)
\left(\delta n^2 r\bar{r}+\delta^2 n r^2\right).$$
\comment{ 
Finally,
$$\left(1\pm o(1)\right)\delta n^2 rr'\le 2\left(1+o(1)\right)\delta n^2 r^2$$
by Lemma~\ref{inneraverageandoverallaverage}.
}
\end{proof}

Recall that the variance of any random variable $X$
equals
$\mathop{\mathrm E}[X^2]-(\mathop{\mathrm E}[X])^2$.

\begin{lemma}\label{boundonthevarianceofthelengthofthematching}
If
$\delta nr\ge \Delta$,
then
$$
\mathop{\mathrm{var}}_{\bs{\pi}}\left(
\sum_{i=1}^{\lfloor n/2\rfloor}\,
d\left(\bs{\pi}\left(2i-1\right),\bs{\pi}\left(2i\right)\right)
\right)
\le
\left(1+
o(1)
\right)\delta n^2 r^2.
$$
\end{lemma}
\begin{proof}
By equations~(\ref{startofthemeanofsquare})--(\ref{endofthemeanofsquare})
and
Lemmas~\ref{distinctdistancesproduct}--\ref{thehardestparttoboundlemma},
\begin{eqnarray*}
&&\mathop{\mathrm E}_{\bs{\pi}}\left[\,
\left(\sum_{i=1}^{\lfloor n/2\rfloor}\,
d\left(\bs{\pi}\left(2i-1\right),\bs{\pi}\left(2i\right)\right)\right)^2
\,\right]\\
&\le&
\frac{1}{4}\cdot \left(1
+
O\left(\frac{1}{n}\right)
\right) n^2
\bar{r}^2
+\left(
\frac{1}{2}
+
O\left(\frac{1}{n}\right)
\right)\left(\delta n^2 r\bar{r}+\delta^2 nr^2\right).
\end{eqnarray*}
This and
Lemma~\ref{expectedmatchingsize} imply
\begin{eqnarray*}
&&\mathop{\mathrm{var}}_{\bs{\pi}}\left(
\sum_{i=1}^{\lfloor n/2\rfloor}\,
d\left(\bs{\pi}\left(2i-1\right),\bs{\pi}\left(2i\right)\right)
\right)\\
&\le&
O\left(\frac{1}{n}\right)
\cdot
n^2
\bar{r}^2
+\left(
\frac{1}{2}
+
O\left(\frac{1}{n}\right)
\right)\left(\delta n^2 r\bar{r}+\delta^2 nr^2\right).
\end{eqnarray*}
Finally, invoke Lemma~\ref{inneraverageandoverallaverage} and recall that
$\delta=\omega(1/\sqrt{n})$.
\end{proof}

\begin{lemma}\label{randommatchingconcentrationlemma}
If
$\delta nr \ge \Delta$,
then
{\small 
\begin{eqnarray*}
\Pr_{\bs{\pi}}\left[
\,
\left|
\,
\left(
\sum_{i=1}^{\lfloor n/2\rfloor}
d\left(\bs{\pi}\left(2i-1\right),\bs{\pi}\left(2i\right)\right)
\right)
-
\frac{1}{2}\cdot\left(1\pm
O\left(\frac{1}{n}\right)
\right)n\bar{r}
\,
\right|
\ge
k\sqrt{\left(1+
o(1)
\right)\delta}\, nr
\,
\right]
\le
\frac{1}{k^2}
\end{eqnarray*}
for all $k>1$.
}
\end{lemma}
\begin{proof}
Use
Chebyshev's inequality
and
Lemmas~\ref{expectedmatchingsize}~and~\ref{boundonthevarianceofthelengthofthematching}.
\end{proof}


\comment{ 
\section{Estimating the average distance on a graph}

This section presents an efficient estimation of
\begin{eqnarray}
\bar{r}\equiv \frac{1}{n^2}\cdot \sum_{x,y\in[n]}\,d\left(x,y\right)
\label{averagedistanceingeneral}
\end{eqnarray}
when $d$ is a graph metric.
As
in Sec.~\ref{approximationratiosection}, take $\epsilon=1/n^{o(1)}$ less
than a small constant and $\delta$ as in line~1 of {\sf Las Vegas median}.
Note that
the proofs of
Lemmas~\ref{pointsoutofballarebad}--\ref{innerballmediantotaldistance}
are independent from the choice of $z$ in line~3 of {\sf Las Vegas median}.
In particular, they remain to hold when
{\sf pick point} in line~3
returns a uniformly random point in $[n]$.
}

\comment{ 
\begin{lemma}\label{averagedistancecannotbetoolargewhp}
$$\Pr_{\bs{u}}\left[\,r\ge 9\bar{r}\,\right]\le\frac{1}{9}.$$
\end{lemma}
\begin{proof}
By lines~2--3
of {\sf average distance},
$$
\mathop{\mathrm{E}}_{\bs{u}}\left[\,r\,\right]
=\frac{1}{n}\cdot \sum_{u\in[n]}\,\frac{1}{n}\cdot \sum_{x\in[n]}\,
d\left(u,x\right)
\stackrel{\text{(\ref{averagedistance})}}{=}
\bar{r}.
$$
Now use Markov's inequality.
\end{proof}
}

\begin{lemma}\label{thekeyconcentration}
If
$\delta nr\ge \Delta$,
then
{\footnotesize 
$$
\Pr_{\bs{\pi}}\left[\,
\left|\,
\left(
\sum_{i=1}^{\lfloor n/2\rfloor}\,
d\left(\bs{\pi}\left(2i-1\right),\bs{\pi}\left(2i\right)\right)
\right)
-
\frac{1}{2}\cdot\left(1\pm
O\left(\frac{1}{n}\right)
\right)n\bar{r}
\,\right|
\ge
k\sqrt{\left(1+o(1)
\right)\delta}\, n\bar{r}
\,\right]
\le\frac{1}{k^2}
$$
}
for all $k>1$.
\end{lemma}
\begin{proof}
By
inequalities~(\ref{bestisnoworsethanaverage})~and~(\ref{theaveragedistancefromanoptimalpoint}),
$$r\le\bar{r}.$$
This and
Lemma~\ref{randommatchingconcentrationlemma}
complete the proof.
\end{proof}

We now arrive at an efficient estimation of the average distance on a graph.

\begin{theorem}\label{theoremonestimatinggraphaveragedistance}
Given $n\in\mathbb{Z}^+$,
$\epsilon=\omega(1/n^{1/4})$
and oracle access to a graph metric
$d\colon[n]\times[n]\to\mathbb{N}$,
a real number in $[\,(1-\epsilon)\bar{r},(1+\epsilon)\bar{r}\,]$
can be found in $O(n)$ time with probability
$1/2+\Omega(1)$.
\end{theorem}
\begin{proof}
Let $G=([n],E)$ be an undirected unweighted graph inducing the distance function $d$.
Then pick $x$, $y\in [n]$ with
$d(x,y)=\Delta$, i.e., $(x,y)$ is a furthest pair of vertices of $G$.
Find
a simple
shortest $x$-$y$ path, denoted
$(v_0=x,v_1,\ldots,v_\Delta=y)$, in $G$.
By equation~(\ref{theaveragedistancefromanoptimalpoint}),
\begin{eqnarray}
r
\ge
\frac{1}{n}\cdot\sum_{i=0}^\Delta\,d\left(p^*,v_i\right).
\label{totaldistancetoverticesonalongestpath}
\end{eqnarray}
Now,
{\small 
\begin{eqnarray}
\sum_{i=0}^\Delta\,d\left(p^*,v_i\right)
=
\frac{1}{2}\cdot
\sum_{i=0}^\Delta
d\left(p^*,v_i\right)+d\left(p^*,v_{\Delta-i}\right)
\ge\frac{1}{2}\cdot\sum_{i=0}^\Delta\, d\left(v_i,v_{\Delta-i}\right)
=\frac{1}{2}\cdot \sum_{i=0}^\Delta\, \left|\,\Delta-2i\,\right|
\ge \frac{\Delta^2}{4},
\label{totaldistancestopointsonapath}
\end{eqnarray}
}
where the first inequality (resp., the second equality) follows from
the triangle inequality (resp.,
$(v_0,v_1,\ldots,v_\Delta)$
being
a shortest $v_0$-$v_\Delta$ path).\footnote{It is easy to verify that $\sum_{i=0}^\Delta\,|\,\Delta-2i\,|
=(\Delta+2)\Delta/2$
if
$\Delta\equiv 0\pmod{2}$
and $\sum_{i=0}^\Delta\,|\,\Delta-2i\,|
=(\Delta+1)^2/2$ otherwise.}
By
inequalities~(\ref{totaldistancetoverticesonalongestpath})--(\ref{totaldistancestopointsonapath}),
\begin{eqnarray}
nr
\ge\frac{\Delta^2}{4}.
\label{totaldistancetoverticesonalongpath}
\end{eqnarray}
Because $d$ is a graph metric, $d(x,y)\ge1$ for all distinct $x$, $y\in[n]$.
So
by equation~(\ref{theaveragedistancefromanoptimalpoint}),
\begin{eqnarray}
r
\ge
\frac{1}{n}\cdot \sum_{x\in[n]\setminus\{p^*\}}\,1\ge\frac{1}{2}
\label{trivialaveragedistancelowerbound}
\end{eqnarray}
for all $n\ge2$.

By
inequalities~(\ref{totaldistancetoverticesonalongpath})--(\ref{trivialaveragedistancelowerbound}),
$$
\delta nr
\ge
\delta
\cdot\max\left\{\frac{\Delta^2}{4},\frac{n}{2}\right\}.
$$
So
\begin{eqnarray}
\delta nr
\ge \Delta
\label{thesmallballwillbetheuniverse}
\end{eqnarray}
for all
sufficiently large
$n$.\footnote{If $\Delta\ge 4/\delta$, then
$\delta\Delta^2/4\ge \Delta$.
Otherwise, $\delta n/2\ge \Delta$ for all
$n>8/\delta^2$.
Finally, recall that $\delta=\omega(1/\sqrt{n})$.}
By
equation~(\ref{thedeltavalue}),
\begin{eqnarray}
3\sqrt{\left(1+o(1)\right)\delta}\le 0.1\,\epsilon
\label{theerrortermcalculated}
\end{eqnarray}
for all sufficiently large $n$.
By
inequalities~(\ref{thesmallballwillbetheuniverse})--(\ref{theerrortermcalculated}),
Lemma~\ref{thekeyconcentration} with
$k=3$
and recalling that $\epsilon=\omega(1/n^{1/4})$,
\begin{eqnarray}
\Pr_{\bs{\pi}}\left[\,
\sum_{i=1}^{\lfloor n/2\rfloor}\,d\left(\bs{\pi}(2i-1),\bs{\pi}(2i)\right)
\in \left[\,\left(\frac{1}{2}-\frac{\epsilon}{2}\right)n\bar{r},
\left(\frac{1}{2}+\frac{\epsilon}{2}\right)n\bar{r}
\,\right]
\,\right]\ge 1-\frac{1}{9}
\label{concentrationofmatchingconcluded}
\end{eqnarray}
for
all sufficiently
large $n$.
Consequently, the output of line~2 of {\sf average distance}
in Fig.~\ref{mainalgorithm}
is in $[\,(1-\epsilon)\bar{r},(1+\epsilon)\bar{r}\,]$
with probability $1/2+\Omega(1)$.
Line~1 takes $O(n)$ time by the Knuth shuffle.
Clearly,
line~2
also takes
$O(n)$ time.
\end{proof}

The time complexity of $O(n)$ in Theorem~\ref{theoremonestimatinggraphaveragedistance}
is independent of $\epsilon$.
But
for general metrics,
we do not know
whether the time complexity of $O(n/\epsilon^2)$
in
Fact~\ref{averagedistancepriorresult}
can be improved to $O(n/\epsilon^{2-\Omega(1)})$.

\comment{ 
Pick $\tilde{z}$ from $[n]$ uniformly at random.
Analogous to line~4 of {\sf Las Vegas median}, define
\begin{eqnarray}
\tilde{r}\equiv \frac{1}{n}\cdot \sum_{x\in[n]}\,d\left(\tilde{z},x\right).
\tilde{r}\equiv
\end{eqnarray}
Observe that the proofs of
Lemmas~\ref{pointsoutofballarebad}--\ref{randommatchingconcentrationlemma}
are independent from the picking of $z$.
So
Lemma~\ref{randommatchingconcentrationlemma}
remains to hold with $z$ and $r$ replaced by $\tilde{z}$ and $\tilde{r}$, respectively.
That is,
}


\appendix
\section{Analyzing {\sc max square sum}}\label{analyzingthemaximizationproblem}

\setcounter{theorem}{0}
\numberwithin{theorem}{section}

{\sc Max square sum} has an optimal solution, denoted
$\{\tilde{d}_{x,y}\in\mathbb{R}\}_{x,y\in [n]}$,
because
its
feasible solutions
(i.e., those satisfying
constraints~(\ref{averagedistanceconstraint})--(\ref{largestdistanceconstraint}))
form a
closed and bounded
subset of
$\mathbb{R}^{(n^2)}$.
(Recall from elementary mathematical analysis that a continuous
real-valued function on a
closed and bounded
subset of $\mathbb{R}^k$
has a maximum value, where $k<\infty$.)
Note that
$\{\tilde{d}_{x,y}\in\mathbb{R}\}_{x,y\in [n]}$
must be feasible to {\sc max square sum}.
Below is a consequence of constraint~(\ref{averagedistanceconstraint}).

\begin{lemma}\label{maximumnumberoflargestvaluevariables}
\begin{eqnarray}
\left|
\left\{
\left(x,y\right)\in
[n]\times[n]
\mid \tilde{d}_{x,y}=\delta nr
\right\}
\right|
\le \left\lfloor\frac{n\bar{r}}{\delta r}\right\rfloor.
\label{maximumnumberoflargestvaluevariablesinequality}
\end{eqnarray}
\end{lemma}
\begin{proof}
Clearly,
$$
n^2 \bar{r}
\stackrel{\text{(\ref{averagedistanceconstraint})}}{=}
\sum_{x,y\in [n]}\,\tilde{d}_{x,y}
\ge
\left|
\left\{
\left(x,y\right)\in
[n]\times[n]
\mid \tilde{d}_{x,y}=\delta nr
\right\}
\right|
\cdot \delta nr.
$$
Furthermore, the left-hand side of
inequality~(\ref{maximumnumberoflargestvaluevariablesinequality})
is an integer.
\end{proof}

\begin{lemma}\label{supportofoptimalsolution}
$$
\left|
\left\{
\left(x,y\right)\in
[n]\times[n]
\mid \tilde{d}_{x,y}>0
\right\}
\right|
\le \left\lfloor\frac{n\bar{r}}{\delta r}\right\rfloor+1.
$$
\end{lemma}
\begin{proof}
Assume otherwise.
Then
\begin{eqnarray*}
&&\left|
\left\{
\left(x,y\right)\in
[n]\times[n]
\mid \left(\tilde{d}_{x,y}>0\right)\land\left(\tilde{d}_{x,y}\neq \delta nr\right)
\right\}
\right|\\
&\ge&
\left|
\left\{
\left(x,y\right)\in
[n]\times[n]
\mid \tilde{d}_{x,y}>0
\right\}
\right|
-\left|
\left\{
\left(x,y\right)\in
[n]\times[n]
\mid \tilde{d}_{x,y}=\delta nr
\right\}
\right|\\
&\ge&
\left\lfloor\frac{n\bar{r}}{\delta r}\right\rfloor+2
-\left|
\left\{
\left(x,y\right)\in
[n]\times[n]
\mid \tilde{d}_{x,y}=\delta nr
\right\}
\right|\\
&\stackrel{\text{Lemma~\ref{maximumnumberoflargestvaluevariables}}}{\ge}&
2.
\end{eqnarray*}
So by
constraint~(\ref{largestdistanceconstraint}) (and the feasibility of
$\{\tilde{d}_{x,y}\}_{x,y\in [n]}$ to {\sc max square sum}),
$$
\left|
\left\{
\left(x,y\right)\in
[n]\times[n]
\mid 0<\tilde{d}_{x,y}<\delta nr
\right\}
\right|
\ge 2.
$$
Consequently,
there exist distinct $(x',y')$,
$(x'',y'')\in [n]\times[n]$
satisfying
\begin{eqnarray}
0<\tilde{d}_{x',y'},\, \tilde{d}_{x'',y''}<\delta nr.
\label{thenonfullvariable1}
\end{eqnarray}
By symmetry, assume
$\tilde{d}_{x',y'}\ge \tilde{d}_{x'',y''}$.
By
inequality~(\ref{thenonfullvariable1}),
there exists a small real number $\beta>0$
such that
increasing $\tilde{d}_{x',y'}$ by $\beta$ and simultaneously
decreasing $\tilde{d}_{x'',y''}$ by $\beta$
will preserve
constraints~(\ref{averagedistanceconstraint})--(\ref{largestdistanceconstraint}).
I.e., the solution $\{\hat{d}_{x,y}\in\mathbb{R}\}_{x,y\in [n]}$ defined below is
feasible
to
{\sc max square sum}:
\begin{eqnarray}
\hat{d}_{x,y}=
\left\{
\begin{array}{ll}
\tilde{d}_{x',y'}+\beta, & \text{if $(x,y)=(x',y')$},\\
\tilde{d}_{x'',y''}-\beta, & \text{if $(x,y)=(x'',y'')$},\\
\tilde{d}_{x,y},&\text{otherwise}.
\end{array}
\right.\label{variatedsolution}
\end{eqnarray}

Clearly,
objective~(\ref{objectiveofoptimization})
w.r.t.\ $\{\hat{d}_{x,y}\}_{x,y\in[n]}$
exceeds that w.r.t.\
$\{\tilde{d}_{x,y}\}_{x,y\in [n]}$
by
\begin{eqnarray*}
&&\left\lfloor\frac{n}{2}\right\rfloor\cdot
\frac{1}{n\cdot(n-1)}
\cdot
\sum_{x,y\in [n]}\,\left({\hat{d}}^2_{x,y}-{\tilde{d}}^2_{x,y}
\right)\\
&\stackrel{\text{(\ref{variatedsolution})}}{=}&
\left\lfloor\frac{n}{2}\right\rfloor\cdot
\frac{1}{n\cdot(n-1)}
\cdot
\left(
\left(\tilde{d}_{x',y'}+\beta\right)^2+
\left(\tilde{d}_{x'',y''}-\beta\right)^2
-{\tilde{d}}^2_{x',y'}-{\tilde{d}}^2_{x'',y''}
\right)\\
&=&
\left\lfloor\frac{n}{2}\right\rfloor\cdot
\frac{1}{n\cdot(n-1)}
\cdot
\left(
2\beta\tilde{d}_{x',y'}-2\beta\tilde{d}_{x'',y''}
+2\beta^2
\right)\\
&>&0,
\end{eqnarray*}
where the inequality holds
because $\tilde{d}_{x',y'}\ge \tilde{d}_{x'',y''}$ and
$\beta>0$.

In summary,
$\{\hat{d}_{x,y}\}_{x,y\in [n]}$ is a feasible solution
to {\sc max square sum}
achieving a greater
objective~(\ref{objectiveofoptimization})
than the optimal solution
$\{\tilde{d}_{x,y}\}_{x,y\in[n]}$ does, a contradiction.
\end{proof}

We now
bound the optimal value of
{\sc max square sum}.

\begin{theorem}
The optimal value of {\sc max square sum}
is at most
$$
\left\lfloor\frac{n}{2}\right\rfloor\cdot
\frac{1}{n\cdot(n-1)}
\cdot
\left(
\left\lfloor
\frac{n\bar{r}}{\delta r}
\right\rfloor+1
\right)
\cdot \left(\delta nr\right)^2
$$
\end{theorem}
\begin{proof}
W.r.t.\
the optimal (and thus feasible) solution
$\{\tilde{d}_{x,y}\}_{x,y\in [n]}$,
objective~(\ref{objectiveofoptimization}) equals
\begin{eqnarray*}
&&\left\lfloor\frac{n}{2}\right\rfloor\cdot
\frac{1}{n\cdot(n-1)}
\cdot
\sum_{x,y\in [n]}\,
\chi\left[\tilde{d}_{x,y}\neq 0\right]\cdot {\tilde{d}}^2_{x,y}\\
&\stackrel{\text{(\ref{largestdistanceconstraint})}}{\le}&
\left\lfloor\frac{n}{2}\right\rfloor\cdot
\frac{1}{n\cdot(n-1)}
\cdot
\sum_{x,y\in [n]}\,
\chi\left[\tilde{d}_{x,y}>0\right]\cdot \left(\delta nr\right)^2,
\end{eqnarray*}
where $\chi[P]=1$ if $P$ is true and $\chi[P]=0$ otherwise, for any
predicate $P$.
Now invoke Lemma~\ref{supportofoptimalsolution}.
\end{proof}

\comment{ 
By the triangle inequality,
\begin{eqnarray*}
&&\sum_{\text{distinct $u$, $v\in B(z,\delta nr)$}}\,
d\left(u,v\right)\\
&\le&
\sum_{\text{distinct $u$, $v\in B(z,\delta nr)$}}\,
\left(d\left(z,u\right)+d\left(z,v\right)\right)\\
&=&
\left(\left|B(z,\delta nr)\right|-1\right)\cdot
\sum_{u\in B(z,\delta nr)}\,
d\left(z,u\right)
+\left(\left|B(z,\delta nr)\right|-1\right)\cdot
\sum_{v\in B(z,\delta nr)}\,
d\left(z,v\right).
\end{eqnarray*}
}

\bibliographystyle{plain}
\bibliography{las_vegas_median_merged}

\noindent

\end{document}